%% file: Submission.tex
\documentclass[conference]{IEEEtran}

\ifCLASSINFOpdf
\else
\fi

\usepackage{microtype}
\usepackage{graphicx}
\usepackage{subfigure}
\usepackage{booktabs,pifont} 

\usepackage{empheq}

\usepackage{pgfplots}
\pgfplotsset{compat=1.12}
\usetikzlibrary{shapes,arrows}
\usetikzlibrary{positioning}
\usepackage{tikz}
\usetikzlibrary{positioning,chains,fit,shapes,calc}
\usepackage{amsmath,bm,times}
\usetikzlibrary{calc}
\usepackage{dsfont}
\usepackage{cuted}
\usepackage{caption}

\usepackage{mathtools}
\DeclarePairedDelimiter{\ceil}{\lceil}{\rceil}
\DeclarePairedDelimiter{\floor}{\lfloor}{\rfloor}
\usepackage{amsthm}
\usepackage{comment}
\usepackage{enumitem}
\usepackage{arydshln}
\usepackage{cite}
\usepackage{multirow}
\usepackage{calc}
\usepackage{blkarray}
\usepackage{url}
\theoremstyle{definition}
\newtheorem{theorem}{Theorem}

\newtheorem{lemma}{Lemma}
\newtheorem{claim}{Claim}
\newtheorem{proposition}{Proposition}
\newtheorem{corollary}{Corollary}
\newtheorem{example}{Example}
\newtheorem{remark}{Remark}
\newtheorem{definition}{Definition}
\usepackage{graphicx}
\usepackage{dblfloatfix}
\usetikzlibrary{decorations.pathreplacing}
\usepackage{blindtext, graphicx, amsfonts,
	amssymb,multirow,epstopdf}
\def\BibTeX{{\rm B\kern-.05em{\sc i\kern-.025em b}\kern-.08em
    T\kern-.1667em\lower.7ex\hbox{E}\kern-.125emX}}

\makeatletter
\renewcommand*\env@matrix[1][*\c@MaxMatrixCols c]{%
  \hskip -\arraycolsep
  \let\@ifnextchar\new@ifnextchar
  \array{#1}}
\makeatother
\usepackage[linesnumbered,ruled]{algorithm2e}

\newcommand{\calB}{\mathcal{B}}
\newcommand{\calC}{\mathcal{C}}
\newcommand{\calA}{\mathcal{A}}

\newcommand{\calW}{\mathcal{W}}

\newcommand{\calG}{\mathcal{G}}
\newcommand{\calH}{\mathcal{H}}
\newcommand{\calI}{\mathcal{I}}

\newcommand{\calN}{\mathcal{N}}
\newcommand{\calO}{\mathcal{O}}

\newcommand{\bfr}{\mathbf{r}}

\newcommand{\bfH}{\mathbf{H}}

\newcommand{\bfA}{\mathbf{A}}

\newcommand{\bfS}{\mathbf{S}}

\newcommand{\bfx}{\mathbf{x}}

\newcommand{\bfB}{\mathbf{B}}

\newcommand{\bfv}{\mathbf{v}}

\newcommand{\bfe}{\mathbf{e}}

\begin{document}

\title{Preserving Sparsity and Privacy in Straggler-Resilient Distributed Matrix Computations}

\definecolor{mygr}{rgb}{0.6,0.4,0.0}
\definecolor{my1color}{rgb}{0.6,0.4,0.0}
\definecolor{mycolor1}{rgb}{0.00000,0.44700,0.74100}%
\definecolor{mycolor2}{rgb}{0.85000,0.32500,0.09800}%
\definecolor{mycolor3}{rgb}{0.45000,0.62500,0.19800}%
\tikzset{
block/.style    = {draw, thick, rectangle, minimum height = 2em, minimum width = 2em},
sum/.style      = {draw, circle, node distance = 1cm},
sum1/.style      = {draw, circle, minimum size = 1.1 cm},
input/.style    = {coordinate},
output/.style   = {coordinate},
}

\author{\IEEEauthorblockN{Anindya Bijoy Das\IEEEauthorrefmark{1}, Aditya Ramamoorthy\IEEEauthorrefmark{2}, David J. Love\IEEEauthorrefmark{1}, Christopher G. Brinton\IEEEauthorrefmark{1}} 
\IEEEauthorblockA{\IEEEauthorrefmark{1}School of Electrical and Computer Engineering, Purdue University, West Lafayette, IN 47907 USA\\
\IEEEauthorrefmark{2}Department of Electrical and Computer Engineering, Iowa State University, Ames, IA 50010 USA\\
\texttt{das207@purdue.edu, adityar@iastate.edu, djlove@purdue.edu, cgb@purdue.edu}
}

}

\IEEEtitleabstractindextext{%

\begin{abstract}

Existing approaches to distributed matrix computations involve allocating coded combinations of submatrices to worker nodes, to build resilience to stragglers and/or enhance privacy. In this study, we consider the challenge of preserving input sparsity in such approaches to retain the associated computational efficiency enhancements. First, we find a lower bound on the weight of coding, i.e., the number of submatrices to be combined to obtain coded submatrices to provide the resilience to the maximum possible number of stragglers (for given number of nodes and their storage constraints). Next we propose a distributed matrix computation scheme which meets this exact lower bound on the weight of the coding. Further, we develop controllable trade-off between worker computation time and the privacy constraint for sparse input matrices in settings where the worker nodes are honest but curious. Numerical experiments conducted in Amazon Web Services (AWS) validate our assertions regarding straggler mitigation and computation speed for sparse matrices.
\end{abstract}

}

\maketitle
\IEEEdisplaynontitleabstractindextext
\IEEEpeerreviewmaketitle
\section{Introduction}
\label{sec:intro}
Computing platforms are constantly stressed to meet the growing demands of end users for data processing. The increasing complexity of data tasks, such as deep neural network AI/ML models, and the sheer volumes of data to be processed, continue to hinder scalability. 

Matrix computations serve as the fundamental building blocks for many data processing tasks in AI/ML and optimization. As data sizes increase, these computations involve high-dimensional matrices, requiring larger runtimes with all else constant. The underlying concept behind distributed computation is to break down the entire operation into smaller tasks and distribute them across multiple worker nodes. However, in these distributed systems, the overall execution time of a job can be significantly affected by slower or failed worker nodes, commonly known as ``stragglers'' \cite{ramamoorthyDTMag20}. 

Recently, a number of coding theory techniques \cite{lee2018speeding, das2019random, dutta2016short, yu2017polynomial, c3les, yu2020straggler, tandon2017gradient, dasunifiedtreatment,  9513242, das2023jsait_submitted, 8849468, 8919859} have been proposed to mitigate the effect of stragglers. A simple example is presented in \cite{lee2018speeding} to illustrate a technique for computing $\bfA^T \bfx$ using three workers. The technique involves partitioning the matrix $\bfA$ into two block-columns, denoted as $\bfA = [\bfA_0 | \bfA_1]$. The workers are then assigned specific tasks: one computes $\bfA_0^T \bfx$, another computes $\bfA_1^T \bfx$, and the third computes $(\bfA_0 + \bfA_1)^T \bfx$. Each worker then handles only half of the computational load, the system can recover $\bfA^T \bfx$ if any two out of the three workers return their results. This means that the system is resilient to the failure or delay of one straggler. In general, the recovery threshold is an important metric defined as the minimum number of workers ($\tau$) required to complete their tasks, enabling the recovery of $\bfA^T \bfx$ from any subset of $\tau$ worker nodes. 

While there are several works that achieve the optimal recovery threshold \cite{yu2017polynomial, 8849468, 8919859, das2019random} for given number of nodes and storage constraints, they possess certain limitations. Real-world datasets, utilized in various domains such as optimization, deep learning, power systems, computational fluid dynamics etc. often consist of sparse matrices. An efficient exploitation of this sparsity can significantly decrease the overall time required for matrix computations \cite{wang2018coded}. However, techniques based on MDS codes \cite{yu2017polynomial, 8849468, 8919859, das2019random} construct dense linear combinations of submatrices; this eliminates the inherent sparsity in the matrix structure. As a consequence, the computation speed of worker nodes can be severely reduced. In this work, one of our objectives is to develop approaches that combine a relatively small number of submatrices while maintaining an optimal recovery threshold.

Another significant issue in distributed computation is the information leakage of the associated ``input'' matrix \cite{tandon2018secure, aliasgari2020private, hollanti2022secure, yu2021entangled}. The assumption is that the input matrix $\bfA$ is known to the central node, but the assigned smaller tasks should involve a protection against information leakage at the worker nodes. Several works \cite{tandon2018secure, aliasgari2020private, hollanti2022secure} propose adding random matrices to the linear combinations of submatrices introduced by MDS codes with the goal of reducing the mutual information between the assigned encoded submatrices and the original matrix $\bfA$. This is again problematic for sparse matrices since the addition of dense random matrices can destroy the sparsity. Thus, we also aim to develop codes that optimize the trade-off between privacy and efficiency.

In this work, first we formulate the problem (Sec. \ref{sec:probform}) and find a lower bound on the number of submatrices to be combined (Sec. \ref{sec:opt_weight}) for coded submatrices that will provide resilience to the maximum number of stragglers in a given system. Next, we develop a novel approach for distributed matrix-vector multiplication (Sec. \ref{sec:prop_approach}) which meets that lower bound, maximizing sparsity preservation while providing resilience to the maximum number of stragglers. Our proposed approach involves a computationally efficient process to find a ``good'' set of random coefficients that make the system numerically stable. Our approach also addresses the privacy issue through a controllable trade-off between privacy leakage and worker computation time for sparse input matrices (Sec. \ref{sec:private}). Finally, we carry out experiments on an Amazon Web Services (AWS) which verify the effectiveness of our proposed methodology compared with baseline approaches in terms of different time, stability, and privacy metrics (Sec. \ref{sec:numexp}). 

\section{Problem Formulation}
\label{sec:probform}
In this work, we examine a distributed system comprising $n$ worker nodes. The primary objective of this system is to calculate the product $\bfA^T \bfx$, where $\bfA \in \mathbb{R}^{t \times r}$ represents a sparse matrix and $\bfx \in \mathbb{R}^{t}$ denotes a vector. It is assumed that the workers are identical in terms of their memory capacity and computational speed. Specifically, each worker can store $\gamma_A = \frac{1}{k_A}$ fraction of the whole matrix $\bfA$, and also, the entire vector $\bfx$. In practical situations, stragglers may arise due to variations in computational speed or failures experienced by certain assigned workers at specific times \cite{das2019random}.

In line with previous approaches, our initial step involves partitioning matrix $\bfA$ into $k_A$ distinct block-columns. Subsequently, we will distribute to each worker node a random linear combination of certain block-columns from $\bfA$ along with the vector $\bfx$. Nevertheless, as discussed in Sec. \ref{sec:intro}, assigning dense linear combinations could lead to the loss of inherent sparsity in the corresponding matrices. To avoid this issue, our goal is to allocate linear combinations involving a smaller number of submatrices \cite{dasunifiedtreatment, das2023distributedisit}. In order to quantify this approach, we introduce the concept of ``weight'' for the encoded submatrices. This measure serves as a crucial metric when dealing with sparse matrices in distributed computations.

\begin{definition}
\label{def:weight}
We define the ``weight'' $(\omega_A)$ of the submatrix encoding procedure as the number of submatrices that are linearly combined to obtain each encoded submatrix. We assume homogeneous weights of the encoded submatrices across the worker nodes, i.e., every node will be assigned linear combinations of the same number of uncoded submatrices.
\end{definition}

Thus, our goal is to obtain the optimal recovery threshold ($\tau = k_A$) while maintaining $\omega_A$ (for the assigned encoded submatrices) as low as possible. We also consider the privacy implications of our approach assuming that the worker nodes are honest but curious.

\section{Minimum Weight of Coding}
\label{sec:opt_weight}
We consider a coded matrix-vector multiplication scheme with homogeneous weight, $\omega_A$, where matrix $\bfA$ is partitioned into $k_A$ disjoint block-columns, $\bfA_0, \bfA_1, \bfA_2, \dots, \bfA_{k_A - 1}$. Now we state the following proposition which provides a lower bound on $\omega_A$ for any coded matrix-vector multiplication scheme with resilience to $s = n - k_A$ stragglers.

\begin{proposition}
\label{prop:lowerbound}
Consider a coded matrix-vector multiplication scheme aiming at resilience to $s = n - k_A$ stragglers out of $n$ total nodes each of which can store $1/k_A$ fraction of matrix $\bfA$. Any scheme that partitions $\bfA$ into $k_A$ disjoint block-columns has to maintain a minimum homogeneous weight $\lceil{\frac{(n-s)(s+1)}{n}}\rceil$.
\end{proposition}
\begin{proof}
Since the scheme aims at resilience to {\it any} $s$ stragglers, any scheme needs to ensure the presence of any $\bfA_i$ (where $i = 0, 1, \dots, k_A - 1$) in at least $s+1$ different nodes. In other words, $\bfA_i$ has to participate within the encoded submatrices in at least $s+1$ different nodes. Now, we assume homogeneous weight $\omega_A$, i. e., each of these $n$ nodes is assigned a linear combination of $\omega_A$ uncoded submatrices from $\bfA$. Thus, we can say $n \; \omega_A \geq k_A (s+1)$, hence,
\begin{align*}
\omega_A \geq {\frac{(n-s)(s+1)}{n}}.
\end{align*} Thus, the minimum homogeneous weight, $\hat{\omega}_A = \Bigl\lceil{\frac{(n-s)(s+1)}{n}}\Bigr\rceil$.
\end{proof}

Now we state the following corollary (of Proposition \ref{prop:lowerbound}) which considers different values of $k_A$ in terms of $s$, and provides the corresponding optimal weights for coded sparse matrix-vector multiplication.

\begin{corollary} 
\label{cor:lowerbounds}
Consider the same setting as Prop. \ref{prop:lowerbound} for coded matrix-vector multiplication. Now, 
\begin{itemize}
    \item (i) if $k_A > s^2$, then $\hat{\omega}_A = s + 1$.
    \item (ii) if $s \leq k_A \leq s^2$, then $\ceil{\frac{s+1}{2}} \leq \hat{\omega}_A \leq s$.
\end{itemize}
\end{corollary}

\begin{proof}
Since $n = k_A + s$, from Prop. \ref{prop:lowerbound}, we have 
\begin{align}
\label{eq:omega_bound}
    \hat{\omega}_A = \Bigl\lceil{\frac{k_A(s+1)}{k_A + s}}\Bigr\rceil = \Bigl\lceil{\frac{1 + s}{1 + \frac{s}{k_A}}}\Bigr\rceil ;
\end{align}hence, $\hat{\omega}_A$ is a non-decreasing function of $k_A$ for fixed $s$. 

\noindent {\bf Part (i):} When $k_A > s^2$, we have $\frac{s}{k_A} < \frac{1}{s}$, and $\frac{1 + s}{1 + \frac{s}{k_A}} > \frac{1 + s}{1 + \frac{1}{s}} = s$. Thus, from \eqref{eq:omega_bound}, $\hat{\omega}_A > s$. In addition, from \eqref{eq:omega_bound}, for any $s \geq 0$, we have $\hat{\omega}_A \leq s + 1$. Thus, we have $\hat{\omega}_A = s + 1$.

\noindent {\bf Part (ii):} If $k_A = s^2$, from \eqref{eq:omega_bound}, we have $\hat{\omega}_A = s$. Similarly, if $k_A = s$, from \eqref{eq:omega_bound}, we have $\hat{\omega}_A = \ceil{\frac{s+1}{2}}$. Thus, the non-decreasing property of $\hat{\omega}_A$ in terms of $k_A$ concludes the proof.
\end{proof}

Now we describe a motivating example below where the encoding scheme meets the lower bound mentioned in Prop. \ref{prop:lowerbound}.

\begin{example}
\label{ex:toy_matvec}
\input{matrixvector_n_6_new}
Consider a toy system with $n = 6$ worker nodes each of which can store $1/4$ fraction of matrix $\bfA$. We partition matrix $\bfA$ into $k_A = 4$ disjoint block-columns, $\bfA_0, \bfA_1, \bfA_2, \bfA_3$. According to Prop. \ref{prop:lowerbound}, the optimal weight $\omega_A$ can be as low as $\Bigl\lceil\frac{k_A(s+1)}{k_A + s}\Bigr\rceil = 2$. Now, we observe that the way the jobs are assigned in Fig. \ref{matvec6_new} meets that lower bound, where random linear combinations of $\omega_A = 2$ submatrices are assigned to the nodes. It can be verified that this system has a recovery threshold $\tau = k_A = 4$, and thus, it is resilient to any $s = 2$ stragglers.
\end{example}

\section{Proposed Approach}
\label{sec:prop_approach}

In this section, we detail our overall approach for distributed matrix-vector multiplication which is outlined in Alg. \ref{Alg:New_matvec}. We partition matrix $\bfA$ into $k_A$ block columns, $\bfA_0, \bfA_1, \bfA_2, \dots, \bfA_{k_A - 1}$, and assign a random linear combination of $\omega_A$ (weight) submatrices of $\bfA$ to every worker node. We show that for given $n$ and $k_A$, our proposed approach provides resilience to maximum number of stragglers, $s = n - k_A$. In addition, our coding scheme maintains the minimum weight of coding as mentioned in Prop. \ref{prop:lowerbound}.

Formally, we set $\omega_A = \Bigl\lceil{\frac{k_A(s+1)}{k_A + s}}\Bigr\rceil$, and assign a linear combination of $\bfA_i, \bfA_{i + 1}, \bfA_{i + 2}, \dots, \bfA_{i+\omega_A - 1} \, \left(\textrm{indices modulo} \, k_A \right)$ to worker node $W_i$, for $i = 0, 1, 2, \dots, k_A-1$, where the linear coefficients are chosen randomly from a continuous distribution. Next, we assign a random linear combination of $\bfA_{i\omega_A}, \bfA_{i\omega_A + 1}, \bfA_{i\omega_A + 2}, \dots, \bfA_{(i+1)\omega_A - 1} \, \left(\textrm{indices modulo} \, k_A \right)$ to worker node $W_{i}$, for $i = k_A, k_A + 1, \dots, n-1$. Note that every worker node also receives the vector $\bfx$. Once the fastest $\tau = k_A$ worker nodes finish and return their computation results, the central node decodes $\bfA^T \bfx$. Note that we assume $k_A \geq s$, i.e., at most {\it half} of the nodes may be stragglers.

\begin{algorithm}[t]
	\caption{Proposed scheme for distributed matrix-vector multiplication}
	\label{Alg:New_matvec}
   \SetKwInOut{Input}{Input}
   \SetKwInOut{Output}{Output}
   \vspace{0.1 in}
   \Input{Matrix $\bfA$, vector $\bfx$, $n$-number of workers, $s$-number of stragglers, storage fraction $\gamma_A = \frac{1}{k_A}$, such that $k_A \geq s$.}
   Partition $\bfA$ into $k_A$ disjoint block-columns\;
   Set weight $\, \omega_A =\Bigl\lceil\frac{k_A(s+1)}{k_A + s}\Bigr\rceil$\;
   \For{$i\gets 0$ \KwTo $n-1$}{
   \eIf{$i < k_A$}
   {
   Define $T = \left\lbrace i, i+1, \dots, i + \omega_A - 1 \right\rbrace$ (reduced modulo $k_A$)\;
   }
   { 
   Define $T = \left\lbrace i \omega_A, i \omega_A + 1, \dots, (i+1)\omega_A - 1 \right\rbrace$ (reduced modulo $k_A$)\;
   }
   Create a random vector $\bfr$ of length $k_A$ with entries  $r_{m}$, $0\leq m \leq k_A - 1$\;
   Create a random linear combination of $\bfA_{q}$'s where $q \in T$, thus $\tilde{\bfA}_i = \sum\limits_{q \in T} r_{q} \bfA_q$\;
   Assign encoded submatrix $\tilde{\bfA}_i$ and the vector $\bfx$ to worker node $W_i$\;
   }
   \Output{The central node recovers $\bfA^T \bfx$ from the returned results by the fastest $k_A$ nodes.}
   \vspace{0.1 in}
\end{algorithm}

\subsection{Straggler Resilience Guarantee}
Next we state the following lemma which would assist us to prove Theorem \ref{thm:matvec} which discusses straggler resilience of our proposed scheme.

\begin{lemma}
\label{lem:hall}
Choose any $m \leq k_A$ worker nodes out of all $n$ nodes in the distributed system. Now, if we assign the jobs to the worker nodes according to Alg. \ref{Alg:New_matvec}, the total number of participating uncoded $\bfA$ submatrices within those $m$ worker nodes is lower bounded by $m$. 
\end{lemma}

\begin{proof}
First we partition all $n$ worker nodes into {\it two} sets where the first set, $\calW_0$ includes the first $k_A$ nodes and the second set, $\calW_1$, includes the next $s$ worker nodes, i.e., we have 
\begin{align*}
    \calW_0 &= \left\lbrace W_0, W_1, W_2, \dots, W_{k_A - 1} \right\rbrace ; \\
\textrm{and} \;\; \;  \calW_1 &= \left\lbrace W_{k_A}, W_{k_A+1}, \dots, W_{n-1}  \right\rbrace . 
\end{align*} Thus, we have $|\calW_0| = k_A$ and $|\calW_1| = s \leq k_A$. Now, we choose any $m \leq k_A$ worker nodes, where we choose $m_0$ nodes from $\calW_0$ and $m_1$ nodes from $\calW_1$, so that $m = m_0 + m_1$. We denote set of the participating uncoded $\bfA$ submatrices within those nodes as $\calA_0$ and $\calA_1$, respectively. Hence, to prove the lemma, we need to show $|\calA_0 \cup \calA_1| \geq m$, for any $m \leq k_A$. 

First, according to Alg. \ref{Alg:New_matvec}, we assign a random linear combination of $\bfA_i, \bfA_{i + 1}, \bfA_{i + 2}, \dots, \bfA_{i+\omega_A - 1} \, \left(\textrm{indices modulo} \, k_A \right)$ to worker node $W_i \in \calW_0$. Thus, the participating submatrices are assigned in a cyclic fashion \cite{das2020coded}, and the total number of participating submatrices within any $m_0$ nodes of $\calW_0$ is  
\begin{align}
\label{eq:m1}
|\calA_0| \geq \min (m_0 + \omega_A - 1, k_A).
\end{align} Next, we state the following claim for the number of participating submatrices in $\calW_1$, with the proof in Appendix \ref{app:proofclaim1}.

\begin{claim}
\label{clm:m1gw1}
Choose any $m_1 \geq \omega_A$ nodes from $\calW_1$. The number of participating submatrices within these nodes, $|\calA_1| = k_A$. 
\end{claim}

Now, if $m_1 \leq \omega_A - 1$, from \eqref{eq:m1} we have 
\begin{align*}
    |\calA_0 \cup \calA_1| \geq |\calA_0| & =  \min (m_0 + \omega_A - 1, k_A) \\ 
    & \geq  \min (m_0 + m_1, k_A) \geq m ,
\end{align*} since $m = m_0 + m_1 \leq k_A$. And, if $m_1 \geq \omega_A$, from Claim \ref{clm:m1gw1} we can say,
\begin{align*}
    |\calA_0 \cup \calA_1| \geq |\calA_1| =  k_A \geq m,
\end{align*} which concludes the proof of the lemma.
\end{proof}

\begin{example}
Consider the same scenario in Example \ref{ex:toy_matvec}, where $k_A = 4$ and $s = 2$, therefore, $\calW_0 = \{W_0, W_1, W_2, W_3\}$ and $\calW_1 = \{ W_4, W_5 \}$. Now, choose $m = 3$ nodes, $W_0, W_1$ and $W_4$. Thus, $m_0 = 2$ and $m_1 = 1$. Now, from the figure, we have $\calA_0 = \{\bfA_0, \bfA_1, \bfA_2\}$ and $\calA_1 = \{\bfA_0, \bfA_1\}$. Hence, $|\calA_0 \cup \calA_1| = 3 \geq m$. Similar properties can be shown for any choice $m \leq k_A = 4$ different nodes.
\end{example}

Now we state the following theorem which provides the guarantee of resilience to maximum number of stragglers for given storage constraints.

\begin{theorem}
\label{thm:matvec}
Assume that a system has $n$ worker nodes each of which can store $1/k_A$ fraction of matrix $\bfA$ and the whole vector $\bfx$ for the distributed matrix-vector multiplication $\mathbf{A}^T \mathbf{x}$. If we assign the jobs according to Alg. \ref{Alg:New_matvec}, we achieve resilience to $s = n - k_A$ stragglers.
\end{theorem}

\begin{proof}
According to Alg. \ref{Alg:New_matvec}, first we partition matrix $\bfA$ into $k_A$ disjoint block-columns. Thus, to recover the matrix-vector product, $\bfA^T \bfx$, we need to decode all $k_A$ vector unknowns, $\bfA^T_0 \bfx, \bfA^T_1 \bfx, \bfA^T_2 \bfx, \dots, \bfA^T_{k_A - 1} \bfx$. We denote the set of these $k_A$ unknowns as $\calB$. Now we choose an arbitrary set of $k_A$ worker nodes each of which corresponds to an equation in terms of $\omega_A$ of those $k_A$ unknowns. Denoting the set of $k_A$ equations as $\calC$, we can say,  $|\calB| = |\calC| = k_A$. 

\input{bipartite_2}
Now we consider a bipartite graph $\calG = \calC \cup \calB$, where any vertex (equation) in $\calC$ is connected to some vertices (unknowns) in $\calB$ which participate in the corresponding equation. Thus, each vertex in $\calC$ has a neighborhood of cardinality $\omega_A$ in $\calB$. An example with $k_A = 5$ and $\omega_A = 3$ is shown in Fig. \ref{fig:hall_bipartite}. 

Our goal is to show that there exists a perfect matching among the vertices of $\calC$ and $\calB$. To do so, we consider $\bar{\calC} \subseteq \calC$, where $|\bar{\calC}| = m \leq k_A$. Now, we denote the neighbourhood of $\bar{\calC}$  as $\calN (\bar{\calC}) \subseteq \calB$. Thus, according to Lemma \ref{lem:hall}, for any $m \leq k_A$, we can say that $|\calN (\bar{\calC})| \geq m$. So, according to Hall's marriage theorem \cite{marshall1986combinatorial}, we can say that there exists a perfect matching among the vertices of $\calC$ and $\calB$.

Next we consider the largest matching where the vertex $c_i \in \calC$ is matched to the vertex $b_j \in \calB$, which indicates that $b_j$ participates in the equation corresponding to $c_i$. Now, considering $k_A$ equations and $k_A$ unknowns, we construct the $k_A \times k_A$ coding (or decoding) matrix $\bfH$ where row $i$ corresponds to the equation associated to $c_i$ where $b_j$ participates. We replace row $i$ of $\bfH$ by $\bfe_j$ where $\bfe_j$ is a unit row-vector of length $k_A$ with the $j$-th entry being $1$, and $0$ otherwise. Thus we have a $k_A \times k_A$ matrix where each row has only one non-zero entry which is $1$. In addition, since we have a perfect matching, $\bfH$ will have only one non-zero entry in every column. Thus, $\bfH$ is a permutation of the identity matrix, and therefore, $\bfH$ is full rank. Since the matrix is full rank for a choice of definite values, according to Schwartz-Zippel lemma \cite{schwartz1980fast}, the matrix continues to be full rank for random choices of non-zero entries. Thus, the central node can recover all $k_A$ unknowns from any set of $k_A$ worker nodes.
\end{proof}

\begin{example}
\label{ex:12_3}
\input{matrixvector_opt_n_12}
Consider a system with $n = 12$ nodes each of which can store $1/9$-th fraction of matrix $\bfA$. We partition $\bfA$ as $\bfA_0, \bfA_1, \dots, \bfA_8$. According to Alg. \ref{Alg:New_matvec}, we set the weight $\omega_A = \Bigl\lceil\frac{k_A(s+1)}{k_A + s}\Bigr\rceil = 3$, and assign random linear combinations of $\omega_A$ submatrices to each node as shown in Fig. \ref{matvec12_opt}. It can be verified that $\bfA^T \bfx$ can be recovered from {\it any} $\tau = k_A = 9$ nodes, therefore, the scheme is resilient to {\it any} $s = 3$ stragglers.
\end{example}

\begin{remark}
\label{rem:betterthanjsait}
While our proposed approach meets the lower bound on the weight as mentioned in Prop. \ref{prop:lowerbound}, the approach in \cite{das2023jsait_submitted} assigns a weight $\min(s+1, k_A)$ which can often be higher than ours (e.g., Examples \ref{ex:toy_matvec} and \ref{ex:12_3}), and thus, may lead to reduction in worker computation speed. 
\end{remark}

\subsubsection{Computational Complexity for a Worker Node} 
\label{sec:compcomplexity}
In this work, we assume that the ``input'' matrix, $\bfA \in \mathbb{R}^{t \times r}$, is sparse, i.e., most of the entries of $\bfA$ are zero. Let us assume that the probability for any entry of $\bfA$ to be non-zero is $\mu$, where $\mu > 0$ is very small. According to Alg. \ref{Alg:New_matvec}, we combine $\omega_A$ submatrices (of size $t \times r/k_A$) to obtain the coded submatrices and assign them to the worker nodes. Hence, the probability for any entry of any coded submatrix to be non-zero is $1 - (1 - \mu)^{\omega_A}$ which can be approximated by $\omega_A \mu$. Thus, in our approach, the per worker node computational complexity is $\calO \left( \omega_A \mu \times  \frac{rt}{k_A} \right)$ where $\omega_A = \Bigl\lceil\frac{k_A(s+1)}{k_A + s}\Bigr\rceil$.

On the other hand, the dense coded approaches \cite{yu2017polynomial,8849468, 8919859} combine $k_A$ submatrices for encoding, hence, their per worker node computational complexity is $\calO \left( k_A \mu \times  \frac{rt}{k_A} \right) = \calO \left( \mu \times  r t \right)$ which is $\frac{k_A}{\omega_A} \approx \frac{s + k_A}{ s + 1}$ times higher than that of ours. Moreover, the recent sparse matrix computations approach in \cite{das2023jsait_submitted} combines $s+1$ submatrices for encoding (when $s < k_A$). Thus, its corresponding computational complexity is $\calO \left( (s+1) \mu \times  \frac{rt}{k_A} \right)$; approximately $(1 + s/k)$ times higher than that of ours. We clarify this with the following example.

\begin{example}
Consider the same setting in Example \ref{ex:12_3} where $n = 12$, $k_A = 9$ and $s = 3$. In this scenario, the recent work \cite{das2023jsait_submitted} assigns random linear combinations of $\min(s+1, k_A) = 4$ submatrices to each node. Thus, our proposed approach enjoys a $25\%$ decrease in computational complexity, which could significantly enhance the overall computational speed.
\end{example}


\subsubsection{Numerical Stability and Coefficient Determination Time}
\label{sec:trialtime}
In this section, we discuss the numerical stability of our proposed distributed matrix computations scheme. The condition number is widely regarded as a significant measure of numerical stability for such a system \cite{das2019random, 8849468, 8919859}. In the context of a system consisting of $n$ workers and $s$ stragglers, the worst-case condition number ($\kappa_{worst}$) is defined as the highest condition number among the decoding matrices when considering all possible combinations of $s$ stragglers. In methods involving random coding like ours, the idea is to generate random coefficients multiple (e.g., 20) times  and selecting the set of coefficients that results in the lowest $\kappa_{worst}$ among those trials.

In our proposed method, we partition matrix $\bfA$ into $k_A$ disjoint block-columns, which underscores the necessity to recover $k_A$ vector unknowns. Consequently, in each attempt, we must determine the condition numbers of ${n \choose k_A}$ decoding matrices, each of size $k_A \times k_A$. This whole process has a total complexity of $\calO\left( {n \choose k_A} k_A^3\right)$. On the other hand, the recent sparse matrix computation techniques, such as sparsely coded straggler (SCS) optimal scheme discussed in \cite{das2020coded} or the class-based scheme discussed in \cite{dasunifiedtreatment} partition matrix $\bfA$ into $\Delta_A = \textrm{LCM}(n, k_A)$ block-columns. Thus, in each attempt, they need to ascertain the condition numbers of ${n \choose k_A}$ matrices, each of which has a size $\Delta_A \times \Delta_A$, resulting in a total complexity of $\calO\left( {n \choose k_A} \Delta_A^3\right)$. Since $\Delta_A$ can be considerably larger than $k_A$, those methods involve significantly more complexity compared to our proposed scheme. For instance, if we consider a scenario where $n$ and $k_A$ are co-prime, then $\Delta_A = n k_A$, and thus the complexity of the approaches presented in \cite{das2020coded, dasunifiedtreatment} is approximately $\calO\left( n^3 \right)$ times higher than our method.


\subsection{Private Matrix-vector Multiplication}
\label{sec:private}
Now, we discuss how we can modify Alg. \ref{Alg:New_matvec} to add protection against information leakage of the ``input'' matrix $\bfA$ in the worker nodes, which we assume are honest but curious. The traditional idea developed in several private distributed computations approaches \cite{aliasgari2020private, tandon2018secure} is to add dense random matrices to the submatrices of the ``input'' matrix. While this can provide protection against information leakage up to certain levels, it substantially increases the number of non-zero entries in the encoded submatrices of an originally sparse input matrix, which can reduce the overall computation speed. 

In our scheme, we propose that the central node will generate a sparse matrix $\bfS \in \mathbb{F}^{t \times r/k_A}$ where the probability of any entry being non-zero is $\mu$. Next, the central node will add $\bfS$ to all the encoded submatrices to be assigned to the worker nodes according to Alg. \ref{Alg:New_matvec}. In other words, if the central node was supposed to send the encoded submatrix $\tilde{\bfA}_i$ to worker node $W_i$ according to Alg. \ref{Alg:New_matvec}, then for private sparse matrix computations, the central node will send $\bar{\bfA}_i = \tilde{\bfA}_i + \bfS$ to worker node $W_i$. The upcoming corollary proves that the central node can recover the final result, $\bfA^T \bfx$ from any $k_A + 1$ nodes (in a similar process as in Sec. \ref{sec:prop_approach}). Note that the central node sends the vector $\bfx$ to all $n$ nodes. The overall procedure for private matrix-vector multiplication is outlined in Alg. \ref{Alg:sec_matvec}.

\begin{algorithm}[t]
	\caption{Proposed scheme for Private distributed matrix-vector multiplication for non-colluding nodes}
	\label{Alg:sec_matvec}
   \SetKwInOut{Input}{Input}
   \SetKwInOut{Output}{Output}
   \Input{Matrix $\bfA \in \mathbb{F}^{t \times r}$, vector $\bfx \in \mathbb{F}^{t \times 1}$, $n$-number of nodes, storage fraction $\frac{1}{k_A}$, where $n > k_A$.}
   Create a sparse random matrix $\bfS \in \mathbb{F}^{t \times r/k_A}$, where the probability of any entry to be non-zero is $\mu$\;
   Create a random vector $\bfr$ of length $n$\;
   \For{$i\gets 0$ \KwTo $n-1$}{
   Create encoded submatrix $\tilde{\bfA}_i$ according to Alg. \ref{Alg:New_matvec}\;
   Assign submatrix $\bar{\bfA}_i = \tilde{\bfA}_i + \bfr_i \bfS$ to worker $W_i$\;
   Assign vector $\bfx$  to worker node $W_i$\;
   }
   \Output{The central node recovers $\bfA^T \bfx$ from the returned results by the fastest $k_A + 1$ nodes.}
\end{algorithm}

\begin{corollary}
\label{cor:private}
Assume that a system has $n$ worker nodes each of which can store $1/k_A$ fraction of matrix $\bfA$ for conducting private matrix-vector multiplication $\mathbf{A}^T \bfx$. If we assign the jobs according to Alg. \ref{Alg:sec_matvec} to achieve our desired level of protection against information leakage of $\bfA$, we achieve resilience to $s = n - (k_A+1)$ stragglers.    
\end{corollary}

\begin{proof}
We prove the corollary in a similar fashion as we have proved Theorem \ref{thm:matvec}. Instead of $k_A$ vector unknowns, $\bfA^T_0 \bfx, \bfA^T_1 \bfx, \bfA^T_2 \bfx, \dots, \bfA^T_{k_A - 1} \bfx$, to recover $\bfA^T \bfx$, we have one more unknown, $\bfS^T \bfx$ involved in this process. Similar to the proof of Theorem \ref{thm:matvec}, we denote the set of these $k_A+1$ unknowns as $\calB$, and choose an arbitrary set of $k_A+1$ worker nodes each of which corresponds to an equation in terms of $\omega_A+1$ of those $k_A+1$ unknowns. Denoting the set of $k_A+1$ equations as $\calC$, we can say,  $|\calB| = |\calC| = k_A+1$. 

We again consider a bipartite graph $\calG = \calC \cup \calB$, and claim that a perfect matching exists between the vertices in $\calC$ and $\calB$. The reason is that the new unknown $\bfS^T \bfx$ participates in every equation, hence, the size of the of neighborhood of $\bar{\calC} \in \calC$ will always increase by $1$ (as compared to Theorem \ref{thm:matvec}) when $|\bar{\calC}| = m \leq k_A$. Thus, for any $\bar{\calC}$, when $|\bar{\calC}| = m \leq k_A + 1$, the size of the neighborhood $|\calN(\bar{\calC})| \geq m$. This proves the perfect matching, and then, similar to the proof of Theorem \ref{thm:matvec}, using Schwartz-Zippel lemma \cite{schwartz1980fast}, we can prove the corollary.
\end{proof}

We consider a system of non-colluding worker nodes which are honest but curious. In this setting, in order to be private from an information-theoretic standpoint, the encoded matrices $\bar{\bfA}_i$ should not leak any information about the data matrix $\bfA$. In this regard, denote the mutual information of two random variables $X$ and $Y$ as $\calI(X, Y)$. A perfectly private scheme in our setting must satisfy the information-theoretic constraint, $\calI\left(\bar{\bfA}_i, \bfA \right) = 0$, for $i = 0, 1, \dots, n -1$. Denoting $\calH (X, Y)$  as the joint entropy of two random variables $X$ and $Y$, for our scheme we can write
\begin{align*}
    & \calI\left(\bar{\bfA}_i; \bfA \right)  =  \calI \left( \tilde{\bfA}_i + \bfS; \bfA \right) \nonumber \\ 
    = & \calH\left( \tilde{\bfA}_i + \bfS\right) - \calH\left( \tilde{\bfA}_i + \bfS | \bfA \right)  = \calH\left( \tilde{\bfA}_i + \bfS\right) - \calH\left( \bfS | \bfA \right)
\end{align*} Now, for small $\eta$, the number of non-zero entries in any $\tilde{\bfA}_i$ is approximately $\omega_A \eta \times \frac{rt}{k_A}$. Thus, we have
\begin{align}
\label{eq:sec}
    & \; \calI\left(\bar{\bfA}_i; \bfA \right) \nonumber \\
 \approx \; & \;{  \left( \omega_A\eta \; + \mu \; - \omega_A\eta \mu \;\right) \frac{rt}{k_A} \; log |\mathbb{F}|} - {\mu \; \frac{rt}{k_A} \; log |\mathbb{F}|} \nonumber \\
 = \; & \; \omega_A\eta \left( 1 - \mu \right)\; \frac{rt}{k_A} \; log |\mathbb{F}|
\end{align}
Thus, $\calI\left(\bar{\bfA}_i, \bfA \right)$ decreases with the increase of $\mu$; if the central node uses a denser $\bfS$, the system will have more protection, at the expense of longer computation times due to sparsity being destroyed. The system will be fully protected if $\mu = 1$, in other words, when $\bfS$ is fully dense. 

\begin{remark}
A recent work \cite{xhemrishi2022distributed} also studied this privacy issue in sparse matrix computations for a different setting of distributed computation. In that setting, the worker nodes are partitioned into two non-communicating clusters, the untrusted cluster and the partly trusted cluster, and different number of tasks are assigned to different nodes. This objective is different than our focus on being resilient to the maximum number of stragglers.
\end{remark}

\input{num_exp_double.tex}
\section{Conclusion}
\label{sec:conclusion}

In this study, we devised a distributed scheme for multiplying large matrices by vectors, specifically designed for sparse input matrices. First we found a lower bound on the weight for the encoding of any scheme for the resilience to the maximum number of stragglers for given storage constraints. Our proposed straggler-optimal approach meets the lower bound and maintains the inherent sparsity of the input matrix $\bfA$ up to a certain extent. As a result, it substantially reduces both computation and communication delays compared to dense coded methods. We also explored the privacy aspect of sparse matrix computations when the nodes are honest but curious. We achieved a controllable balance between the preserved sparsity level and information leakage. Our claims were corroborated through numerical experiments conducted on an AWS cluster.

A future direction can include developing schemes for sparse distributed matrix-matrix multiplication which meets the lower bound on the weight. Another direction may include developing sparsely coded schemes with protection against information leakage when the worker nodes can collude among them.

\appendix
\subsection{Proof of Claim \ref{clm:m1gw1}}
\label{app:proofclaim1}
\begin{proof}
Consider the worker nodes in $\calW_1$. According to Alg. \ref{Alg:New_matvec}, we assign a linear combination of $\bfA_{i\omega_A}, \bfA_{i\omega_A + 1}, \bfA_{i\omega_A + 2}, \dots, \bfA_{(i+1)\omega_A - 1} \, \left(\textrm{indices modulo} \, k_A \right)$ to worker node $W_{i}$, for $i = k_A, k_A + 1, \dots, n-1$. Thus, the participating submatrices in worker node $W_{k_A}$ are $\bfA_0, \bfA_1, \dots, \bfA_{\omega_A - 1} \left(\textrm{indices reduced modulo} \, k_A \right)$. Similarly, the participating submatrices in $W_{k_A + 1}$ are $\bfA_{\omega_A}, \bfA_{\omega_A+1}, \dots, \bfA_{2\omega_A - 1} \left(\textrm{indices reduced modulo} \, k_A \right)$. In a consequence, $\omega_A$ number of submatrices participate in each of those $s$ worker nodes sequentially in an increasing order in terms of their indices (reduced modulo $k_A$). 

Now, denote the number of appearances of any submatrix $\bfA_i$ within the nodes in $\calW_1$ by $\bfv_i \geq 0$. Thus, for any $0 \leq j, k \leq k_A - 1$, we have $ |\bfv_j  -  \bfv_k| \leq 1$, where $\sum_{i = 0}^{k_A - 1} \bfv_i = s \omega_A$. Thus, the average of these $\bfv_i$'s is $\rho = \frac{s \omega_A}{k_A}$. If $\rho$ is an integer, then $\bfv_i = \floor{\rho} = \rho$ for $i = 0, 1, 2, \dots, k_A - 1$, since for every pair of $j,k$, we have $|\bfv_j  -  \bfv_k| \leq 1$. Similarly, if $\rho$ is not an integer, then $\bfv_i \geq \floor{\rho}$. Thus, within all $s$ nodes of $\calW_1$, every submatrix participates in at least $\floor{\rho}$ times over $\floor{\rho}$ distinct nodes. In other words, any submatrix may not participate in at most $s - \floor{\rho}$ nodes within the nodes of $\calW_1$.

First, consider the case, $k_A = s$. Here, every submatrix participates in $\floor{\rho} = \omega_A$ nodes, therefore, any submatrix does not participate in $s - \omega_A$ nodes. But, we choose any $m_1 \geq \omega_A$ nodes in $\calW_1$, where $\omega_A = \ceil{\frac{s+1}{2}}$, since $k_A = s$. Thus,
\begin{align*}
   2 \omega_A \geq s + 1 > s \;\; \textrm{which indicates that}, \;\; \omega_A > s - \omega_A.
\end{align*}In addition, since $m_1 \geq \omega_A$, we claim that $m_1 > s - \omega_A$. Thus, every submatrix will participate at least once within those chosen $m_1$ nodes, hence $|\calA_1| = k_A$. 

Next, consider the other case when  $k_A > s$. Again, since we choose any arbitrary $m_1 \geq \omega_A$ nodes in $\calW_1$, we are leaving $s - m_1$ nodes in $\calW_1$. But 
\begin{align*}
    s - m_1 \leq s - \omega_A < s - \floor{\rho}.
\end{align*} The second inequality holds since $s < k_A$. Thus, every submatrix will participate at least once within those $m_1 \geq \omega_A$ nodes, hence $|\calA_1| = k_A$.
\end{proof}

%
%


\ifCLASSOPTIONcaptionsoff
  \newpage
\fi

\bibliographystyle{IEEEtran}
\bibliography{citations}
\end{document}

%% file: matrixvector_n_6_new.tex
\begin{figure}[t]
\centering
\definecolor{mycolor6}{rgb}{0.92941,0.69412,0.12549}%

\resizebox{0.99\linewidth}{!}{
\begin{tikzpicture}[auto, thick, node distance=2cm, >=triangle 45]

\draw

    node [sum, minimum size = 0.8cm, fill=blue!30] (blk1) {$W0$}
    node [sum, minimum size = 0.8cm,fill=blue!30,right = 1 cm of blk1] (blk2) {$W1$}
    node [sum, minimum size = 0.8cm,fill=blue!30,right = 1 cm of blk2] (blk3) {$W2$}
    node [sum, minimum size = 0.8cm,fill=blue!30,right = 1 cm of blk3] (blk4) {$W3$}
    node [sum, minimum size = 0.8cm,fill=blue!30,right = 1 cm of blk4] (blk5) {$W4$}
    node [sum, minimum size = 0.8cm,fill=blue!30,right = 1 cm of blk5] (blk6) {$W5$}
    
    node [block, fill=green!30, minimum width = 4.8em, below = 0.5 cm of blk1] (blk11) {$\left\lbrace\bfA_0, \bfA_1 \right\rbrace$}
    node [block, fill=green!30, minimum width = 4.8em, below = 0.5 cm of blk2] (blk21) {$\left\lbrace\bfA_1, \bfA_2 \right\rbrace$}
    node [block, fill=green!30, minimum width = 4.8em, below = 0.5 cm of blk3] (blk31) {$\left\lbrace\bfA_2, \bfA_3 \right\rbrace$}
    node [block, fill=green!30, minimum width = 4.8em, below = 0.5 cm of blk4] (blk41) {$\left\lbrace\bfA_3, \bfA_0 \right\rbrace$}
    node [block, fill=green!30, minimum width = 4.8em, below = 0.5 cm of blk5] (blk51) {$\left\lbrace\bfA_0, \bfA_2\right\rbrace$}
    node [block, fill=green!30, minimum width = 4.8em, below = 0.5 cm of blk6] (blk61) {$\left\lbrace\bfA_1, \bfA_3 \right\rbrace$}
        
;
\draw[->](blk1) -- node{} (blk11);
\draw[->](blk2) -- node{} (blk21);
\draw[->](blk3) -- node{} (blk31);
\draw[->](blk4) -- node{} (blk41);
\draw[->](blk5) -- node{} (blk51);
\draw[->](blk6) -- node{} (blk61);

\end{tikzpicture}
}
\caption{\small Submatrix allocation for a system with $n = 6$, $s = 2$ and $\gamma_A = \frac{1}{4}$ according to Alg. \ref{Alg:New_matvec}. Here, the weight of every coded submatrix is $\omega_A = \Bigl\lceil\frac{k_A(s+1)}{k_A + s}\Bigr\rceil = 2$. Any $\{\bfA_i, \bfA_j\}$ indicates a random linear combination of $\bfA_i$ and $\bfA_j$.}
\label{matvec6_new}
\end{figure}
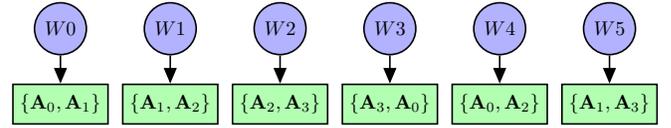 

%% file: bipartite_2.tex
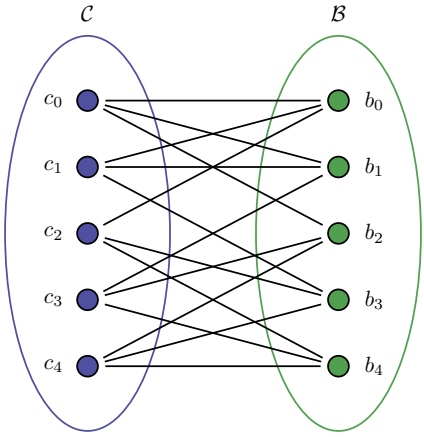
\begin{figure}[t]
\definecolor{myblue}{RGB}{80,80,160}
\definecolor{mygreen}{RGB}{80,160,80}
\centering
\captionsetup{justification=centering}
\resizebox{0.65\linewidth}{!}{

\begin{tikzpicture}[thick,
  every node/.style={draw,circle},
  fsnode/.style={fill=myblue},
  ssnode/.style={fill=mygreen},
  every fit/.style={ellipse,draw,inner sep=-2pt,text width=2cm},
  ->,shorten >= 3pt,shorten <= 3pt
]
\begin{scope}[start chain=going below,node distance=7mm]
\foreach \i in {1,2,...,5}
{
  \pgfmathtruncatemacro{\j}{\i - 1}
  \node[fsnode,on chain] (f\i) [label=left: $c_{\j}$] {};
}
\end{scope}

\begin{scope}[xshift=4cm,start chain=going below,node distance=7mm]
\foreach \i in {6,7,...,10}
{ 
  \pgfmathtruncatemacro{\j}{\i - 6}
  \node[ssnode,on chain] (s\i) [label=right: $b_{\j}$] {};
}
\end{scope}

\node [myblue,fit=(f1) (f5),label=above:$\mathcal{C}$] {};
\node [mygreen,fit=(s6) (s10),label=above:$\mathcal{B}$] {};

\draw[-] (s6) -- (f1);
\draw[-] (s7) -- (f1);
\draw[-] (s8) -- (f1);
\draw[-] (s6) -- (f2);
\draw[-] (s7) -- (f2);
\draw[-] (s9) -- (f2);
\draw[-] (s6) -- (f3);
\draw[-] (s9) -- (f3);
\draw[-] (s10) -- (f3);
\draw[-] (s7) -- (f4);
\draw[-] (s8) -- (f4);
\draw[-] (s10) -- (f4);
\draw[-] (s8) -- (f5);
\draw[-] (s9) -- (f5);
\draw[-] (s10) -- (f5);

\end{tikzpicture}
}
\caption{\small A bipartite graph $\calG = \calC \cup \calB$ with $|\calC| = |\calB| = 5$ where the set of equations is $\calC$ and the set of unknowns is $\calB$. Here, $\omega_A = 3$.}
\label{fig:hall_bipartite}
\end{figure} 

%% file: matrixvector_opt_n_12.tex
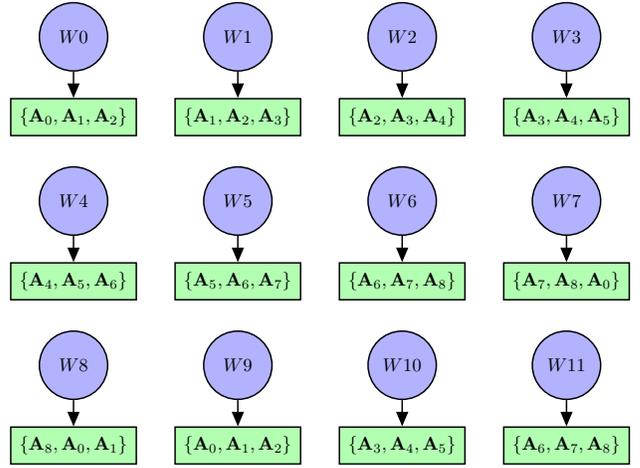
\begin{figure}[t]
\centering
\definecolor{mycolor6}{rgb}{0.92941,0.69412,0.12549}%

\resizebox{0.95\linewidth}{!}{
\begin{tikzpicture}[auto, thick, node distance=2cm, >=triangle 45]

\draw

    node [sum, minimum size = 1.3cm, fill=blue!30] (blk1) {$W0$}
    node [sum, minimum size = 1.3cm,fill=blue!30,right = 1.8 cm of blk1] (blk2) {$W1$}
    node [sum, minimum size = 1.3cm,fill=blue!30,right = 1.8 cm of blk2] (blk3) {$W2$}
    node [sum, minimum size = 1.3cm,fill=blue!30,right = 1.8 cm of blk3] (blk4) {$W3$}
    node [sum, minimum size = 1.3cm,fill=blue!30,below = 1.8 cm of blk1] (blk5) {$W4$}
    node [sum, minimum size = 1.3cm,fill=blue!30,right = 1.8 cm of blk5] (blk6) {$W5$}
    node [sum, minimum size = 1.3cm,fill=blue!30,right = 1.8 cm of blk6] (blk7) {$W6$}
    node [sum, minimum size = 1.3cm,fill=blue!30,right = 1.8 cm of blk7] (blk8) {$W7$}
    node [sum, minimum size = 1.3cm,fill=blue!30,below = 1.8 cm of blk5] (blk9) {$W8$}
    node [sum, minimum size = 1.3cm,fill=blue!30,right = 1.8 cm of blk9] (blk_10) {$W9$}
    node [sum, minimum size = 1.3cm,fill=blue!30,right = 1.8 cm of blk_10] (blk_11) {$W10$}
    node [sum, minimum size = 1.3cm,fill=blue!30,right = 1.8 cm of blk_11] (blk_12) {$W11$}

    node [block, fill=green!30, minimum width = 6.8em, below = 0.5 cm of blk1] (blk11) {$\left\lbrace\bfA_0, \bfA_1, \bfA_2 \right\rbrace$}
    node [block, fill=green!30, minimum width = 6.8em, below = 0.5 cm of blk2] (blk21) {$\left\lbrace\bfA_1, \bfA_2, \bfA_3 \right\rbrace$}
    node [block, fill=green!30, minimum width = 6.8em, below = 0.5 cm of blk3] (blk31) {$\left\lbrace\bfA_2, \bfA_3, \bfA_4 \right\rbrace$}
    node [block, fill=green!30, minimum width = 6.8em, below = 0.5 cm of blk4] (blk41) {$\left\lbrace\bfA_3, \bfA_4, \bfA_5 \right\rbrace$}
    node [block, fill=green!30, minimum width = 6.8em, below = 0.5 cm of blk5] (blk51) {$\left\lbrace\bfA_4, \bfA_5, \bfA_6 \right\rbrace$}
    node [block, fill=green!30, minimum width = 6.8em, below = 0.5 cm of blk6] (blk61) {$\left\lbrace\bfA_5, \bfA_6, \bfA_7 \right\rbrace$}
    node [block, fill=green!30, minimum width = 6.8em, below = 0.5 cm of blk7] (blk71) {$\left\lbrace\bfA_6, \bfA_7, \bfA_8 \right\rbrace$}
    node [block, fill=green!30, minimum width = 6.8em, below = 0.5 cm of blk8] (blk81) {$\left\lbrace\bfA_7, \bfA_8, \bfA_0 \right\rbrace$}
    node [block, fill=green!30, minimum width = 6.8em, below = 0.5 cm of blk9] (blk91) {$\left\lbrace\bfA_8, \bfA_0, \bfA_1 \right\rbrace$}
    node [block, fill=green!30, minimum width = 6.8em, below = 0.5 cm of blk_10] (blk_101) {$\left\lbrace\bfA_0, \bfA_1, \bfA_2 \right\rbrace$}
    node [block, fill=green!30, minimum width = 6.8em, below = 0.5 cm of blk_11] (blk_111) {$\left\lbrace\bfA_3, \bfA_4, \bfA_5 \right\rbrace$}
    node [block, fill=green!30, minimum width = 6.8em, below = 0.5 cm of blk_12] (blk_121) {$\left\lbrace\bfA_6, \bfA_7, \bfA_8 \right\rbrace$}
    
;
\draw[->](blk1) -- node{} (blk11);
\draw[->](blk2) -- node{} (blk21);
\draw[->](blk3) -- node{} (blk31);
\draw[->](blk4) -- node{} (blk41);
\draw[->](blk5) -- node{} (blk51);
\draw[->](blk6) -- node{} (blk61);
\draw[->](blk7) -- node{} (blk71);
\draw[->](blk8) -- node{} (blk81);
\draw[->](blk9) -- node{} (blk91);
\draw[->](blk_10) -- node{} (blk_101);
\draw[->](blk_11) -- node{} (blk_111);
\draw[->](blk_12) -- node{} (blk_121);

\end{tikzpicture}
}
\caption{\small Submatrix allocation for $n = 12$ workers and $s = 3$ stragglers, with $\gamma_A = \frac{1}{9}$ according to Alg. \ref{Alg:New_matvec}. Here, the weight of every submatrix is $\omega_A = \Bigl\lceil\frac{k_A(s+1)}{k_A + s}\Bigr\rceil = 3$. Any $\{\bfA_i, \bfA_j, \bfA_k\}$ indicates a random linear combination of the corresponding submatrices where the coefficients are chosen i.i.d. at random from a continuous distribution.}
\label{matvec12_opt}
\end{figure} 

%% file: num_exp_double.tex
\section{Numerical Experiments}
\label{sec:numexp}

In this section, we evaluate the effectiveness of our proposed approach by conducting numerical experiments and comparing its performance with various competing methods \cite{yu2017polynomial, 8849468, 8919859, das2020coded, dasunifiedtreatment, das2023jsait_submitted}. Note that there are several other works specifically developed for sparse matrix computations. Among them, the approach in \cite{wang2018coded} does not provide resilience to maximum number of stragglers for given storage constraints. The approach in \cite{xhemrishi2022distributed} partitions the worker nodes into untrusted and partly trusted cluster, which is not aligned to our assumption. The approach in \cite{ji2022sparse} assigns some jobs to the central node to reduce the probability of rank-deficiency in the decoding, which is also not in line of our assumptions. So, in the numerical experiment section, we do not consider these approaches.

We explore two different distributed systems: the first one consists of $n = 30$ worker nodes with $s = 5$ stragglers and the other consists of $n = 36$ nodes with $s = 8$ stragglers. We focus on a sparse input matrix $\bfA$ sized $40,000 \times 31,500$ and a dense vector $\bfx$ of length $40,000$. We consider two distinct scenarios in which the sparsity of $\bfA$ is $98\%$, and $99\%$, respectively. This implies that randomly selected $98\%$ and $99\%$ entries, respectively, in the matrix $\bfA$ are zero. It is worth noting that there exist numerous practical instances where data matrices demonstrate such (or, even more) levels of sparsity (refer to \cite{sparsematrices} for specific examples). The experiments are carried out on an AWS (Amazon Web Services) cluster, utilizing a {\tt c5.18xlarge} machine as the central node and {\tt t2.small} machines as the worker nodes.

\vspace{0.05 in}

{\bf Worker computation time:} Table \ref{table:worker_comp} presents a comparison among different methods based on the computation time required by worker nodes to complete their respective tasks. In these scenarios, where $k_A = 25$ or $28$, the approaches described in \cite{yu2017polynomial, 8849468, 8919859} allocate linear combinations of $k_A$ submatrices to the worker nodes. Consequently, the original sparsity of matrix $\bfA$ is lost within the encoded submatrices. As a result, the worker nodes experience a significantly increased processing time for their tasks compared to our proposed approach or the methods outlined in \cite{das2020coded, dasunifiedtreatment, das2023jsait_submitted}, which are specifically designed for sparse matrices and involve smaller weights.

To discuss the effectiveness of our approach in more details, we compare the weight of the coding of our approach against the approach in \cite{das2023jsait_submitted}. In the first scenario, when $n = 30$ and $s = 5$, our approach sets the weight $\Bigl\lceil{\frac{(n-s)(s+1)}{n}}\Bigr\rceil = \Bigl\lceil{\frac{25 \times 6}{30}}\Bigr\rceil = 5$, whereas the approach in \cite{das2023jsait_submitted} uses a weight $\min(s+1, k_A) = \min(6, 25) = 6$. Thus, our approach involves around $17\%$ less computational complexity per worker node, which is supported by the results in Table \ref{table:worker_comp}. Similarly, when $n = 36$ and $s = 8$, our proposed approach involves a weight $\Bigl\lceil{\frac{28 \times 9}{36}}\Bigr\rceil = 7$, which is smaller than the corresponding weight, $s + 1 = 9$, used by the approach in \cite{das2023jsait_submitted}.

\begin{table*}[t]
\caption{{\small Comparison of worker computation time and communication delay (matrix transmission time) for matrix-vector multiplication for $n = 30, s = 5$, and $n = 36, s = 8$, when randomly chosen $98\%$ and $99\%$ entries of matrix $\bfA$ are zero.}}
\vspace{-0.1in}
\label{table:worker_comp}
\begin{center}
\begin{small}
\begin{sc}
\begin{tabular}{c c c c c c c c c c c c c}
\hline
\toprule
\multirow{3}{*}{Methods} & & \multicolumn{5}{c}{$n = 30$ and $s = 5$} & &  \multicolumn{5}{c}{$n = 36$ and $s = 8$} \\ \cline{3-7} \cline{9-13}
 & & \multicolumn{2}{c}{Comp. Time (in ms)} & &  \multicolumn{2}{c}{Comm. Delay (in s)} & & \multicolumn{2}{c}{Comp. Time (in ms)} & & \multicolumn{2}{c}{Comm. Delay (in s)}\\ \cline{3-4} \cline{6-7} \cline{9-10} \cline{12-13}
& & $99\%$ &  $98\%$ & & $99\%$ & $ 98\%$ & & $99\%$ &  $ 98\%$ &  & $ 99\%$ &  $ 98\%$   \\
 \midrule
Poly. Code  \cite{yu2017polynomial} & & $61.4$ & $62.3$ & & $0.67$ & $1.14$ & & $55.7$ &  $56.3$ & & $0.52$ & $0.95$  \\
Ortho Poly  \cite{8849468}   & & $62.2$ & $61.7$ & & $0.69$ & $1.17$ & & $56.2$ &  $56.4$ & & $0.49$ & $0.91$  \\
RKRP Code \cite{8919859} & & $60.3$ & $61.1$ & & $0.65$ & $1.11$ & & $56.8$ &  $57.4$ & & $0.51$ & $0.93$  \\
SCS Opt. Sch. \cite{das2020coded} & & $24.1$ & $38.3$ & & $0.24$ & $0.37$ & & $28.1$ & $41.3$ & & $0.28$ & $0.42$  \\
Class-based \cite{dasunifiedtreatment} & & $17.3$ & $28.2$ & & $0.20$ & $0.31$ & & $22.1$ & $33.7$ & & $0.24$ & $0.35$  \\
Cyclic Code \cite{das2023jsait_submitted} & & $19.5$ & $33.4$ & & $0.23$ & $0.35$ & & $26.7$ &  $37.6$ & & $0.27$ & $0.39$  \\
{\bf Proposed Scheme} & & $\mathbf{16.7}$ & $\mathbf{27.7}$ & & $\mathbf{0.19}$ & $\mathbf{0.32}$ & & $\mathbf{21.8}$ &  $\mathbf{33.9}$ & & $\mathbf{0.24}$ & $\mathbf{0.34}$  \\
\bottomrule
\end{tabular}
\end{sc}
\end{small}
\end{center}
\vspace{-0.15in}
\end{table*}%

\vspace{0.04 in}

{\bf Communication delay:} Table \ref{table:worker_comp} also illustrates the delay incurred during the transmission of encoded submatrices from the central node to the worker node. The approaches presented in \cite{yu2017polynomial}, \cite{8849468}, and \cite{8919859} employ dense linear combinations of submatrices, resulting in a significant increase in the number of non-zero entries within the encoded submatrices. Consequently, transmitting these large number of non-zero entries leads to a substantial communication delay within the system. In contrast, our proposed scheme mitigates this issue by utilizing encoded submatrices formed through linear combinations of only a limited number of uncoded submatrices which significantly reduces the corresponding communication delay.

For example, consider the scenario when $n = 36, s = 8$ and $\bfA$ is $99\%$ sparse. In this scenario, the approach in \cite{yu2017polynomial} needs to transmit up to $0.01 \times 28 \times \frac{40,000 \times 31,500}{28} = 1.26 \times 10^7$ number of non-zero entries to each node. The corresponding number for the approach in \cite{das2023jsait_submitted, das2023distributedisit} is $0.01 \times (s+1) \times \frac{40,000 \times 31,500}{28} = 4.05 \times 10^6$. On the other hand, the corresponding number for our proposed method is $0.01 \times \Bigl\lceil{\frac{k_A(s+1)}{n}}\Bigr\rceil \times \frac{40,000 \times 31,500}{28} = 3.15 \times 10^6$, which is smaller than the previous ones, and clarifies the reduction of communication delay as mentioned in Table \ref{table:worker_comp}.

\vspace{0.04 in}

{\bf Numerical stability:} 
Next, we assess the numerical stability of distributed systems using different coded matrix computation techniques. We examine the condition numbers of the decoding matrices for various combinations of $n$ workers and $s$ stragglers. By comparing the worst-case condition number ($\kappa_{worst}$) across different methods, we present the $\kappa_{worst}$ values in Table \ref{table:kappa}. The polynomial code approach \cite{yu2017polynomial} involves ill-conditioned Vandermonde matrices and demonstrates significant numerical instability, as evidenced by its notably high value of $\kappa_{worst}$. Our proposed approach, among the numerically stable methods, exhibits smaller $\kappa_{worst}$ value compared to the method in \cite{8849468} where the condition numbers increases exponentially in terms of $s = n - k_A$. Note that the approach in \cite{8919859} provides slightly smaller $\kappa_{worst}$ value than ours; however, as mentioned in Table \ref{table:worker_comp}, the worker computation time and the communication delay are significantly higher in that case, since they assign dense linear combinations to the worker nodes. 

\begin{table}[t]
\caption{\small Comparison among different approaches in terms of worst case condition number $\left(\kappa_{worst} \right)$ 
and the corresponding required time for $10$ trials to find a good set of random coefficients}
\vspace{-0.1 in}
\label{table:kappa}
\begin{center}
\begin{small}
\begin{sc}
\begin{tabular}{c c c}
\hline
\toprule
\multirow{2}{*}{Methods}  & $\kappa_{worst}$ for  & Req. time for \;\\
  & $n = 30$, $s = 5$ & $10$ trials\, (in s) \\

 \midrule

\; \; Poly. Code  \cite{yu2017polynomial}   & $1.47 \times 10^{13}$ & $0$ \\
\; \; Ortho-Poly\cite{8849468}    & $1.40 \times 10^8$ & $0$ \\
\; \; RKRP Code\cite{8919859}    & $1.76 \times 10^6$ & $81.84$ \\
\; \; SCS Opt. Sch. \cite{das2020coded}   & $4.68 \times 10^7$ & $1138.6$ \\
\; \; Class based \cite{dasunifiedtreatment} & $7.16\times 10^6$ & $1479.3$ \\
\; \; Cyclic Code  \cite{das2023jsait_submitted} & $1.06 \times 10^7$ & $78.38$ \\
\; \; {\textbf{Prop. Scheme}}  & $\mathbf{{8.21 \times 10^6}}$ & $\mathbf{{77.41}}$ \\
\bottomrule
\end{tabular}
\end{sc}
\end{small}
\end{center}
\vspace{-0.2 in}
\end{table}%

\vspace{0.02 in}

{\bf Coefficient determination time:} 
Next, Table \ref{table:kappa} shows a comparative analysis of various methods with respect to the time required for performing 20 trials to obtain a ``good'' set of random coefficients that ensures numerical stability of the system. As explained in Section \ref{sec:trialtime}, the techniques proposed in \cite{das2020coded} and \cite{dasunifiedtreatment} involve partitioning matrix $\bfA$ into $\Delta_A = \textrm{LCM}(n, k_A)$ block-columns. For instance, when $n = 30$ and $s = 5$, $\Delta_A = 150$ is significantly larger than $k_A = 25$, which denotes the partition level in our approach. Consequently, when dealing with higher-sized matrices to determine the condition number, the methods proposed in \cite{das2020coded} and \cite{dasunifiedtreatment} necessitate considerably more time compared to our approach.


\vspace{0.02 in}

{\bf Trade-off between privacy and worker computation time:} 
Next, we compare the trade-off between protection against information leakage and the worker node computation time. Consider a $99\%$ sparse matrix $\bfA$ of size $40,000 \times 31,500$, i.e., $99\%$ entries of $\bfA$ are zero. We assume the nodes to be honest but curious. Now, according to the discussion in Sec. \ref{sec:private}, we add matrix $\bfS$ to the encoded submatrices of $\bfA$. 
Fig. \ref{privacy_plot} shows the trade-off between the privacy (in terms of $\mu$) and the worker computation time for two different scenarios of $n$ and $s$. The extreme case $\mu = 0$ indicates that the worker node receives only the coded submatrices as outlined by Alg. \ref{Alg:New_matvec}, and in that case, the computation speed is very high. On the other extreme, as clarified in \eqref{eq:sec}, when $\mu = 1$, i.e., dense noise is added to the assigned submatrices, then $\calI\left(\bar{\bfA}_i; \bfA \right) = 0$, which indicates the full protection against information leakage from the honest but curious worker nodes. However, that comes with a sacrifice in the worker node computation speed. In this experiment, we see that the worker computation time is most sensitive at small values of $\mu$, i.e., when less than 20\% non-zero entries are being added. After this point, privacy can be improved with little downside to computational time. Note that the approaches in \cite{das2020coded, dasunifiedtreatment}, while being specifically suited to sparse matrices, do not address the privacy issue.
 
\input{privacy}

%% file: privacy.tex
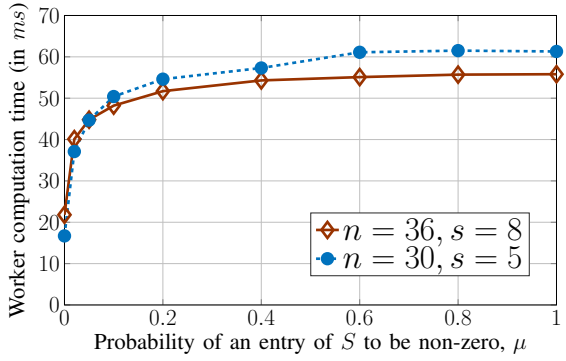
\begin{figure}[t]

\centering
\resizebox{0.85\linewidth}{!}{

\definecolor{mycolor6}{rgb}{0.92941,0.69412,0.12549}%
\definecolor{mycolor7}{rgb}{0.74902,0.00000,0.74902}%
\definecolor{mycolor8}{rgb}{0.60000,0.20000,0.00000}%

\begin{tikzpicture}
\begin{axis}[%
width=5.1in,
height=3.003in,
at={(2.6in,0.85in)},
scale only axis,
xmin=0,
xmax=1,
xlabel style={font=\color{white!15!black}, font=\LARGE},
xlabel={Probability of an entry of $S$ to be non-zero, $\mu$},
ymin=0,
ymax=70,
xtick={0, 0.2, 0.4, 0.6, 0.8, 1},
ytick={0,10,20,30,40,50,60,70},
tick label style={font=\LARGE} ,
ylabel style={font=\color{white!15!black}, font=\LARGE},
ylabel={Worker computation time (in $ms$)},
axis background/.style={fill=white},
xmajorgrids={true},
ymajorgrids={true},
legend style={at={(0.50,0.10)}, nodes={scale=1}, anchor=south west, legend cell align=left, align=left, draw=white!15!black,font = \Huge}
]

\addplot [solid, color=mycolor8, line width=2.0pt, mark=diamond, mark options={solid, mycolor8, scale = 3}]
  table[row sep=crcr]{%
0	21.8\\
0.02	40.1\\
0.05	44.8\\
0.1   48.2\\
0.2	51.7\\
0.4	54.3\\
0.6	55.1\\
0.8	55.7\\
1	55.8\\
};
\addlegendentry{$n = 36, s = 8$}

\addplot [dashed, color=mycolor1, line width=2.0pt, mark=*, mark options={solid, mycolor1, scale = 2}]
  table[row sep=crcr]{%
0	16.7\\
0.02	37.1\\
0.05	44.7\\
0.1	50.4\\
0.2	54.6\\
0.4	57.3\\
0.6	61.1\\
0.8	61.5\\
1	61.3\\
};
\addlegendentry{$n = 30, s = 5$}

\end{axis}
\end{tikzpicture}%
}
\caption{\small Trade-off between the protection against information leakage and the worker computation time. A larger $\mu$ enhances the protection, but reduces the computation speed.}
\vspace{-0.2 in}
\label{privacy_plot}
\end{figure}